\documentclass[11pt]{article}

\usepackage[margin=1in]{geometry}
\usepackage{hyperref, amsmath, amsthm, graphicx, amssymb, CJK}

\newtheorem{lemma}{Lemma}
\newtheorem{theorem}{Theorem}
\newtheorem{corollary}{Corollary}
\newtheorem{proposition}{Proposition}

\theoremstyle{definition}

\newtheorem{assumption}{Assumption}

\DeclareMathOperator{\tr}{tr}
\DeclareMathOperator{\re}{Re}
\DeclareMathOperator{\poly}{poly}

\begin{document}

\title{Finite-size scaling of out-of-time-ordered correlators at late times}

\begin{CJK*}{UTF8}{}

\CJKfamily{gbsn}

\author{Yichen Huang (黄溢辰), Fernando G.S.L. Brand\~ao, and Yong-Liang Zhang (张涌良)\\
Institute for Quantum Information and Matter, California Institute of Technology\\
Pasadena, California 91125, USA\\
yichen.huang@microsoft.com, fbrandao@caltech.edu, ylzhang@caltech.edu
}

\maketitle

\end{CJK*}

\begin{abstract}

Chaotic dynamics in quantum many-body systems scrambles local information so that at late times it can no longer be accessed locally. This is reflected quantitatively in the out-of-time-ordered correlator of local operators, which is expected to decay to zero with time. However, for systems of finite size, out-of-time-ordered correlators do not decay exactly to zero and in this paper we show that the residual value can provide useful insights into the chaotic dynamics. When energy is conserved, the late-time saturation value of the out-of-time-ordered correlator of generic traceless local operators scales as an inverse polynomial in the system size. This is in contrast to the inverse exponential scaling expected for chaotic dynamics without energy conservation. We provide both analytical arguments and numerical simulations to support this conclusion.

\end{abstract}

\section{Introduction}

Nonintegrable quantum many-body systems are expected to exhibit chaotic dynamics, which not only leads to thermalization but also scrambles local information into a nonlocal form. In the Heisenberg picture, the support of $A(t):=e^{iHt}Ae^{-iHt}$ for a local operator $A$ should grow with time under the chaotic dynamics. This growth is reflected in the noncommutativity of $A(t)$ and another local operator $B$ at a different site, which leads to the decay of the out-of-time-ordered correlator (OTOC) $\re\langle A^\dag(t)B^\dag A(t)B\rangle$ \cite{LO69, SS14, RSS15, SS15, RS15, HQRS16, MSS16, RS16, SBSH16, MS16, RY16, GQS17, ZHC19, Kit14, Kit15}. Assume for simplicity that $A$ and $B$ are unitary. Then,
\begin{equation} \label{OTOCdef}
\re\langle A^\dag(t)B^\dag A(t)B\rangle=1-\langle[A(t),B]^\dag[A(t),B]\rangle/2
\end{equation}
so that when the commutator $[\cdots]$ grows, OTOC decays. The chaotic nature of the dynamics is reflected in the fast decay of OTOC away from $1$ in a relatively short time period and the approaching of OTOC to $0$ at late times.

Why does chaotic dynamics lead to such decaying behavior of OTOC? While it is not possible to solve exactly the dynamics of nonintegrable systems in general, we might be able to extract some universal features, at least in certain limits. In a large class of chaotic systems without spatial locality (e.g., large-$N$ theories), OTOC at early time $t$ is given by $1-\epsilon e^{\lambda_Lt}$, where $\epsilon$ is a small prefactor and $\lambda_L$ is a constant. Such an exponential deviation from the initial value is reminiscent of the so-called sensitive dependence on initial conditions (i.e., nearby points in phase space separate from each other over time at an exponential rate) in classical chaos. Thus, $\lambda_L$ may be interpreted as the Lyapunov exponent for quantum systems \cite{Kit14}. In chaotic systems with spatial locality, OTOC of two local operators starts to decay only after a delay that is proportional to the distance between the operators \cite{RSS15, HQRS16, RS16, GQS17, ZHC19}. This is a consequence of the Lieb-Robinson bound \cite{LR72, NS06, HK06}.

In this paper, we study the behavior of OTOC at late times. For simplicity, consider a system of $n$ qubits at infinite temperature so that $\langle\cdots\rangle=\tr(\cdots)/2^n$. In the limit $t\to\infty$, a naive understanding of why OTOC approaches $0$ is as follows. We expand the time-evolved operator in the $n$-qubit Pauli basis $\{\sigma_0=I,\sigma_x,\sigma_y,\sigma_z\}^{\otimes n}$:
\begin{equation} \label{At}
A(t)=\sum_{(k_1,k_2,\ldots,k_n)\in\{0,x,y,z\}^n}a_{k_1k_2\cdots k_n}\sigma_{k_1}\sigma_{k_2}\cdots\sigma_{k_n}.
\end{equation}
The unitarity of $A(t)$ implies
\begin{equation}\label{nor}
\sum_{k_1,k_2,\ldots,k_n}|a_{k_1k_2\cdots k_n}|^2=1.
\end{equation}
After undergoing chaotic evolution for a sufficiently long time, the support of $A(t)$ should be the whole system, and one might expect that the coefficients $a_{k_1k_2\cdots k_n}$ behave like random variables due to the chaotic nature of the dynamics. If we choose $B$ to be the Pauli operator $\sigma_x$ of qubit $1$, then half of the terms in the expansion (\ref{At}) of $A(t)$ commute with $B$ and half of them do not. Thus,
\begin{equation} \label{3}
\langle[A(t),B]^\dag[A(t),B]\rangle=4\sum_{k_2,k_3,\ldots,k_n}|a_{yk_2k_3\cdots k_n}|^2+|a_{zk_2k_3\cdots k_n}|^2\approx4\cdot0.5=2.
\end{equation}
The approximation step follows from Eq. (\ref{nor}) and the fact that we sum over half of the random variables. Substituting Eq. (\ref{3}) into Eq. (\ref{OTOCdef}), we see that OTOC approaches $0$ at late times.

Equation (\ref{At}) with random coefficients is a very simple way to approximate $A(t)$ for large $t$ in chaotic systems and it is oversimplified in some respects. For example, one major difference between this approximation and the exact evolution $A(t) = e^{iHt}Ae^{-iHt}$ is that the latter preserves the spectrum of $A$ while the former does not. How does this discrepancy affect our understanding of the late-time behavior of OTOC? Is it necessary to use more refined and sophisticated approximations in order to fully capture the essence of chaotic dynamics at late times?

We focus on the scaling of late-time OTOC with system size. In finite-size systems, OTOC may converge to a small but finite value, which goes to $0$ when the system size goes to infinity. One might expect this residual value to be exponentially small in the system size because we sum over an exponential number of random variables in Eq. (\ref{3}). However, using a more refined approximation we show that the finite-size scaling of generic late-time OTOC should be inverse polynomial. In fact, the power-law scaling is closely related to energy conservation during the time evolution, which is not captured by simply setting the coefficients in the expansion  (\ref{At}) to be random.

\section{Results} \label{res}

In this section, we introduce basic definitions and provide a summary of results.

Throughout this paper, asymptotic notations are used extensively. Let $f,g:\mathbb R^+\to\mathbb R^+$ be two positive functions. One writes $f(x)=O(g(x))$ if and only if there exist positive numbers $M,x_0$ such that $f(x)\le Mg(x)$ for all $x>x_0$; $f(x)=\Omega(g(x))$ if and only if there exist positive numbers $M,x_0$ such that $f(x)\ge Mg(x)$ for all $x>x_0$; $f(x)=\Theta(g(x))$ if and only if there exist positive numbers $M_1,M_2,x_0$ such that $M_1g(x)\le f(x)\le M_2g(x)$ for all $x>x_0$. To simplify the notation, we use a tilde to hide a polylogarithmic factor, e.g., $\tilde O(f(x)):=O(f(x)\poly\log f(x))$.

For concreteness, consider a chain of $n$ qubits or spin-$1/2$'s with total Hilbert space dimension $d=2^n$ governed by a translationally invariant Hamiltonian $H=\sum_{i=1}^nH_i$, where $H_i$ acts on spins $i,i+1$ (nearest-neighbor interaction). While our discussion is based on a one-dimensional spin system, our results do not rely on the dimensionality of the system or the degrees of freedom being spins. A minor modification of our method leads to similar results in other settings, e.g., fermionic systems in higher dimensions. Assume without loss of generality that $\tr H_i=0$ (traceless) and $\|H_i\|\le1$ (bounded operator norm).

Let $A,B,C,D$ be local (not necessarily unitary) operators with unit operator norm. The residual value of late-time OTOC is
\begin{equation} \label{lateOTOC}
{\rm OTOC}_\infty(A,B,C,D):=\lim_{\tau\to\infty}\frac{1}{\tau}\int_0^\tau\mathrm dt\langle AB(t)CD(t)\rangle,
\end{equation}
where $\langle X\rangle:=\frac{1}{d}\tr X$ denotes the expectation value of an operator at infinite temperature.

Let $\{|1\rangle,|2\rangle,\ldots,|d\rangle\}$ be a complete set of eigenstates of $H$ with corresponding energies $E_1\le E_2\le\cdots\le E_d$ in nondescending order. Let $X_{jk}=\langle j|X|k\rangle$ be the matrix element of an operator in the energy eigenbasis. Define
\begin{equation} \label{eigenconn}
\langle A,B,C,D\rangle_j=(AC)_{jj}B_{jj}D_{jj}+A_{jj}C_{jj}(BD)_{jj}-A_{jj}B_{jj}C_{jj}D_{jj}.
\end{equation}

In strongly chaotic systems, we propose the following formula for late-time OTOC:
\begin{equation} \label{spec}
{\rm OTOC}_\infty(A,B,C,D)\approx\frac{1}{d}\sum_j\langle A,B,C,D\rangle_j.
\end{equation}
Based on this formula, we argue for
\begin{itemize}
\item ${\rm OTOC}_\infty(A,B,A^\dag,B^\dag)$ for traceless local operators $A,B$ vanishes in the thermodynamic limit $n\to\infty$.
\item In finite-size systems, OTOC $\langle AB(t)A^\dag B^\dag(t)\rangle$ saturates to $\Theta(1/n)$ if either $A$ or $B$ (or both) has a finite overlap with the Hamiltonian $H$. We not only derive the prefactor hidden in the big-Theta notation, but also provide a (not necessarily tight) upper bound on the remainder:
\begin{equation} \label{theory}
{\rm OTOC}_\infty(A,B,A^\dag,B^\dag)=\frac{\langle AA^\dag\rangle|\langle HB\rangle|^2+\langle BB^\dag\rangle|\langle HA\rangle|^2}{\langle HH_i\rangle n}+\tilde O(n^{-1.5}).
\end{equation}
\end{itemize}

This is our main result. It is an example where certain properties of quantum chaotic systems can be calculated analytically. For comparison, Table \ref{t1} summarizes the finite-size scaling of late-time OTOC of generic traceless local operators for various types of quantum dynamics.

\begin{table}
\centering
\begin{tabular} {ccc} \hline
types of dynamics & late-time OTOC & references \\ \hline
Haar random unitary & $e^{-\Theta(n)}$ & \cite{Kit16, RY16} \\ 
chaotic Hamiltonian dynamics & $1/\poly n$ & this work \\
many-body localization & $\Theta(1)$ & \cite{HZC17, FZSZ17, Che16, SC17, HL17, CZHF17} \\ \hline
\end{tabular}
\caption{Finite-size scaling of generic late-time OTOC for various types of quantum dynamics.} \label{t1}
\end{table}

The remainder of this paper is organized as follows. In Section \ref{special}, assuming a ``generic'' energy spectrum we present a simple derivation of Eq. (\ref{theory}) for the special case where the local operators in OTOC are terms in the Hamiltonian. In Section \ref{genericSpec}, we extend the approach to the general case using the eigenstate thermalization hypothesis (ETH) \cite{Deu91, Sre94, RDO08}. Thus, we give a rigorous proof of Eqs. (\ref{spec}), (\ref{theory}) based on two very mild assumptions for chaotic systems: a generic spectrum and ETH. In Section \ref{randomU}, we propose a heuristic physical picture for our results from the perspective of interpreting chaotic dynamics with random unitaries. We first introduce a previous approach, which takes into account the unitarity of the dynamics by approximating the time evolution operator $e^{-iHt}$ with a random unitary. Unfortunately, this approximation remains too crude, for it still suggests that the residual value of late-time OTOC is exponentially small in the system size. We show that once energy conservation is also taken into account by requiring the random unitary to act within small energy windows, the finite-size scaling of late-time OTOC becomes inverse polynomial. In Section \ref{numerics}, we support our analytical arguments with numerical simulations of a nonintegrable spin chain. The numerical results suggest that the remainder in Eq. (\ref{theory}) can be improved to $O(n^{-2})$.

\section{Special case} \label{special}

In the case where the local operators in OTOC are terms in the Hamiltonian, we give a simple rigorous proof of Eq. (\ref{theory}) assuming only a generic spectrum.

In strongly chaotic systems, one might expect that the energy spectrum satisfies the ``generic'' condition:
\begin{assumption} [generic spectrum; see, e.g., Ref. \cite{Sre99}] \label{generic}
\begin{equation}
E_p+E_r=E_q+E_s\implies((p=q)~{\rm and}~(r=s))~{\rm or}~((p=s)~{\rm and}~(r=q)).
\end{equation}
\end{assumption}

This assumption is necessary in the sense that it rules out certain integrable (e.g., free-fermion) systems.

Writing out the matrix elements,
\begin{equation}
\langle AB(t)CD(t)\rangle=\frac{1}{d}\sum_{p,q,r,s}A_{pq}B_{qr}C_{rs}D_{sp}e^{i(E_q-E_r+E_s-E_p)t}.
\end{equation}
Substituting into Eq. (\ref{lateOTOC}), we obtain
\begin{equation}
{\rm OTOC}_\infty(A,B,C,D)=\frac{1}{d}\sum_{p,q,r,s}A_{pq}B_{qr}C_{rs}D_{sp}\delta_{E_p+E_r,E_q+E_s},
\end{equation}
where $\delta$ is the Kronecker delta. Assumption \ref{generic} implies
\begin{equation} \label{temp}
{\rm OTOC}_\infty(A,B,C,D)=\frac{1}{d}\sum_{j,k}A_{jj}B_{jk}C_{kk}D_{kj}+\frac{1}{d}\sum_{j,k}A_{jk}B_{kk}C_{kj}D_{jj}-\frac{1}{d}\sum_jA_{jj}B_{jj}C_{jj}D_{jj}.
\end{equation}

Given a Hamiltonian $H$, there are multiple ways to write it as a sum of local terms: $H=\sum_iH_i$. Without loss of generality, we fix this ambiguity by expanding $H$ in the Pauli basis and assigning all Pauli string operators starting at site $i$ to $H_i$ (see Eq. (\ref{hastings}) for an example). This convention implies $\tr(H_jH_k)=0$ for $j\neq k$. Hence, $\langle H_i^2\rangle=\langle HH_i\rangle=\langle H^2\rangle/n$ for any $i$ due to translational invariance. Using this convention,

\begin{theorem}
Assumption \ref{generic} implies
\begin{equation} \label{thmeq}
{\rm OTOC}_\infty(H_1,H_i,H_1,H_i)=2\langle H_i^2\rangle^2/n+O(n^{-2}).
\end{equation}
\end{theorem}

\begin{proof}
We will use the observation that $(H_i)_{jj}=E_j/n$ for any $i$ due to translational invariance. For the present choice of local operators in OTOC, the first term on the right-hand side of Eq. (\ref{temp}) reads
\begin{multline} \label{2}
\frac{1}{d}\sum_{j,k=1}^d(H_1)_{jj}(H_i)_{jk}(H_1)_{kk}(H_i)_{kj}=\frac{1}{dn^2}\sum_{j,k=1}^dE_j\langle j|H_i|k\rangle E_k\langle k|H_i|j\rangle\\
=\frac{1}{dn^2}\tr\left(\sum_{j=1}^d|j\rangle E_j\langle j|H_i\sum_{k=1}^d|k\rangle E_k\langle k|H_i\right)=\frac{\tr(HH_iHH_i)}{dn^2}=\frac{1}{n^2}\sum_{j,k=1}^n\langle H_jH_iH_kH_i\rangle.
\end{multline}
In the last sum, there are $n^2$ terms, most of which are zero because $\tr H_j=\tr H_k=0$. Furthermore, the convention stated above implies $\tr(H_jH_k)=0$ for $j\neq k$. Hence, the number of nonvanishing terms in the last sum of Eq. (\ref{2}) is $n+O(1)$ ($n$ comes from the terms with $j=k$ and $O(1)$ accounts for the remainder). Equation (\ref{2}) equals
\begin{equation}
\frac{1}{n^2}\sum_{j=1}^n\langle H_jH_iH_jH_i\rangle+O(n^{-2})=\langle H_i^2\rangle^2/n+O(n^{-2})+O(n^{-2})=\langle H_i^2\rangle^2/n+O(n^{-2}).
\end{equation}
The second term on the right-hand side of Eq. (\ref{temp}) gives the same result. The last term on the right-hand side of Eq. (\ref{temp}) equals
\begin{equation}
\frac{1}{dn^4}\sum_jE_j^4=\Theta(n^{-2}),
\end{equation}
where we used Eq. (\ref{17}) with $m=4$. This completes the proof.
\end{proof}

\section{General case: Implications of eigenstate thermalization} \label{genericSpec}

In this section, we provide an argument for Eqs. (\ref{spec}), (\ref{theory}). The argument is rigorous assuming a generic spectrum and ETH.

\begin{lemma} [moments]
\begin{align}
&\frac{1}{d}\sum_jE_j^m=\langle H^m\rangle=\Theta(n^{m/2}),\quad\forall~{\rm even~positive~integer}~m, \label{17}\\
&\left|\frac{1}{d}\sum_jE_j^3\right|=|\langle H^3\rangle|=O(n). \label{18}
\end{align}
\end{lemma}

\begin{proof}
Expanding $H$ in the Pauli basis, it suffices to count the number of terms that do not vanish upon taking the trace in the expansion of $H^m$ or $H^3$.
\end{proof}

\begin{lemma} [concentration of eigenvalues] \label{Mar}
Almost all eigenstates have zero energy density:
\begin{equation}
|\{j:|E_j|\ge n^{0.51}\}|/d\le O(n^{-0.01m}),\quad\forall m>0.
\end{equation}
\end{lemma}

\begin{proof}
It follows from Eq. (\ref{17}) and Markov's inequality.
\end{proof}

This lemma allows us to upper bound the total contribution of all eigenstates away from the middle of the spectrum, e.g.,
\begin{equation} \label{tail}
\frac{1}{d}\sum_{j:|E_j|\ge n^{0.51}}E_j^2\le O(n^{2-0.01m}),\quad\forall m>0.
\end{equation}

Lemmas \ref{Mar} and Eqs. (\ref{17}), (\ref{tail}) are related to the fact that $E_j$'s approach a normal distribution in the thermodynamic limit $n\to\infty$ \cite{KLW15, BC15}. Indeed, $|E_j|=\Theta(\sqrt n)$ for almost all $j$.

It suffices to assume ETH for eigenstates in the middle of the spectrum.

\begin{assumption} [eigenstate thermalization hypothesis in the middle of the spectrum] \label{ethasmp}
Let $\delta$ be an arbitrarily small positive constant. For any local operator $X$ with $\|X\|\le1$, there is a function $f_X:[-\delta,\delta]\to[-1,1]$ such that
\begin{equation} \label{asmpeq}
|X_{jj}-f_X(E_j/n)|\le1/\poly n
\end{equation}
for all $j$ with $|E_j|\le\delta n$, where $\poly n$ denotes a polynomial of sufficiently high degree in $n$. We assume that $f_X$ is smooth in the sense of having a Taylor expansion to some low order.
\end{assumption}
It was proposed analytically \cite{Sre99} and supported by numerical simulations \cite{KIH14} that the right-hand side of Eq. (\ref{asmpeq}) can be improved to $e^{-\Omega(n)}$. For our purposes, however, a (much weaker) inverse polynomial upper bound suffices.

\begin{lemma}
For any local operator $X$ and traceless local operator $A$, Assumption \ref{ethasmp} implies
\begin{align}
&f_X(0)=\frac{1}{d}\tr X, \label{zero}\\
&f'_A(0)=\tr(HA)/\tr(HH_i), \label{deriv}\\
&\frac{1}{d}\sum_j|A_{jj}|^2=\frac{|\tr(HA)|^2}{dn\tr(HH_i)}+O(n^{-2}), \label{exp}\\
&\frac{1}{d}\sum_j|A_{jj}|^4=O(n^{-2}). \label{exp2}
\end{align}
\end{lemma}

For a generic traceless local operator $A$, the right-hand side of Eq. (\ref{deriv}) (the normalized overlap between $A$ and the Hamiltonian) is finite and the first term on the right-hand side of Eq. (\ref{exp}) is $\Theta(1/n)$.

\begin{proof} [Proof of Eq. (\ref{zero})]
\begin{equation} \label{zeroc}
\frac{1}{d}\tr X=\frac{1}{d}\sum_jX_{jj}\approx\frac{1}{d}\sum_{j:|E_j|<n^{0.51}}X_{jj}\approx\frac{1}{d}\sum_{j:|E_j|<n^{0.51}}f_X(0)\approx\frac{1}{d}\sum_jf_X(0)=f_X(0),
\end{equation}
where we used Lemma \ref{Mar} in the second and fourth steps. The third step follows from the continuity of $f_X(x)$ at $x=0$. Taking the limit $n\to\infty$, all errors in Eq. (\ref{zeroc}) vanish and thus we obtain Eq. (\ref{zero}). In particular, $f_A(0)=0$ for any traceless local operator $A$.
\end{proof}
\begin{proof} [Proof of Eq. (\ref{deriv})]
\begin{multline} \label{derivc}
\frac{1}{d}\tr(HA)=\frac{1}{d}\sum_jE_jA_{jj}\approx\frac{1}{d}\sum_{j:|E_j|<n^{0.51}}E_jA_{jj}\approx\frac{1}{d}\sum_{j:|E_j|<n^{0.51}}\frac{E_j^2}{n}f_A'(0)\approx\frac{1}{d}\sum_j\frac{E_j^2}{n}f_A'(0)\\
=\tr(HH_i)f_A'(0)/d,
\end{multline}
where we used Lemma \ref{Mar} and Eq. (\ref{tail}) in the second and fourth steps, respectively. In the third step, we used Eq. (\ref{asmpeq}) and the Taylor expansion
\begin{equation} \label{taylor}
f_A(E_j/n)=f_A(0)+f_A'(0)E_j/n+0.5f_A''(0)E_j^2/n^2+O(|E_j|^3/n^3)
\end{equation}
so that the approximation error in this step is upper bounded by
\begin{equation}
\frac{O(1)}{d}\sum_{j:|E_j|<n^{0.51}}\frac{|E_j|^3}{n^2}\le O(n^{-0.47}).
\end{equation}
Taking the limit $n\to\infty$, all errors in Eq. (\ref{derivc}) vanish and thus we obtain Eq. (\ref{deriv}).
\end{proof}
\begin{proof} [Proof of Eq. (\ref{exp})]
\begin{equation}
\frac{1}{d}\sum_j|A_{jj}|^2\approx\frac{1}{d}\sum_{j:|E_j|<n^{0.51}}|A_{jj}|^2\approx\frac{|f'_A(0)|^2}{d}\sum_{j:|E_j|<n^{0.51}}\frac{E_j^2}{n^2}\approx\frac{|f'_A(0)|^2}{d}\sum_j\frac{E_j^2}{n^2}=\frac{|\tr(HA)|^2}{dn\tr(HH_i)},
\end{equation}
where we used Lemma \ref{Mar} and Eqs. (\ref{tail}), (\ref{deriv}) in the first, third, and last steps, respectively. In the second step, we used Eqs. (\ref{asmpeq}), (\ref{taylor}) with the approximation error upper bounded by
\begin{multline} \label{ap4}
\frac{O(1)}{d}\left|\sum_{j:|E_j|<n^{0.51}}\frac{E_j^3}{n^3}\right|+\frac{O(1)}{d}\sum_{j:|E_j|<n^{0.51}}\frac{E_j^4}{n^4}+1/\poly n\approx\frac{O(1)}{d}\left|\sum_j\frac{E_j^3}{n^3}\right|+\frac{O(1)}{d}\sum_j\frac{E_j^4}{n^4}\\
=O(n^{-2})+O(n^{-2})=O(n^{-2}),
\end{multline}
where we used Eqs. (\ref{17}), (\ref{18}).
\end{proof}
\begin{proof} [Proof of Eq. (\ref{exp2})]
\begin{equation}
\frac{1}{d}\sum_j|A_{jj}|^4\approx\frac{1}{d}\sum_{j:|E_j|<n^{0.51}}|A_{jj}|^4\approx\frac{O(1)}{d}\sum_{j:|E_j|<n^{0.51}}\frac{E_j^4}{n^4}\approx\frac{O(1)}{d}\sum_j\frac{E_j^4}{n^4}=O(n^{-2}).
\end{equation}
\end{proof}

Let $J\subseteq\mathbb R$ be an energy interval. Define
\begin{equation}
P_J=\sum_{j:E_j\in J}|j\rangle\langle j|
\end{equation}
as the projector onto $J$.

\begin{lemma} [\cite{AKL16}] \label{trash}
Let $\epsilon<\epsilon'$. For any local operator $X$,
\begin{equation} \label{PXP}
\|P_{(-\infty,\epsilon)}XP_{(\epsilon',\infty)}\|\le\|X\|e^{-\Omega(\epsilon'-\epsilon)}.
\end{equation}
\end{lemma}

This lemma states that local operators cannot (up to an exponentially small error) connect projectors that are far away from each other in the spectrum.

\begin{proof} [Justification of Eq. (\ref{spec})]
Let $c$ be a sufficiently large constant. Consider the first term on the right-hand side of Eq. (\ref{temp}):
\begin{multline} \label{eth}
\frac{1}{d}\sum_{j,k}A_{jj}B_{jk}C_{kk}D_{kj}\approx\frac{1}{d}\sum_j\sum_{k:|E_j-E_k|<c\ln n}A_{jj}B_{jk}C_{kk}D_{kj}\approx\frac{1}{d}\sum_j\sum_{k:|E_j-E_k|<c\ln n}A_{jj}B_{jk}C_{jj}D_{kj}\\
\approx\frac{1}{d}\sum_{j,k}A_{jj}C_{jj}B_{jk}D_{kj}=\frac{1}{d}\sum_jA_{jj}C_{jj}(BD)_{jj},
\end{multline}
where we used Lemma \ref{trash} in the first and third steps: Due to the presence of off-diagonal matrix elements $B_{jk},D_{kj}$, the total contribution of all terms with $|E_j-E_k|\ge c\ln n$ is upper bounded by $1/\poly n$. In the second step of Eq. (\ref{eth}), we replace $C_{kk}$ by $C_{jj}$ using ETH (Assumption \ref{ethasmp}), which states that eigenstates with similar energies have similar local expectation values. A detailed and rigorous error analysis for Eq. (\ref{eth}) is given in Propositions \ref{p1}, \ref{p2} below.

Equation (\ref{eth}) shows that the first term on the right-hand side of Eq. (\ref{temp}) corresponds to the second term on the right-hand side of Eq. (\ref{eigenconn}). Similarly, the second term on the right-hand side of Eq. (\ref{temp}) corresponds to the first term on the right-hand side of Eq. (\ref{eigenconn}). Obviously, the last terms on the right-hand sides of Eqs. (\ref{eigenconn}), (\ref{temp}) are the same. Thus, we obtain Eq. (\ref{spec}).
\end{proof}

\begin{proposition} \label{p1}
The approximation errors in the first and third steps of Eq. (\ref{eth}) are $1/\poly n$, where $\poly n$ denotes a polynomial of sufficiently high degree in $n$.
\end{proposition}

\begin{proof}
Let
\begin{equation}
Q_j=\sum_{k:|E_j-E_k|\ge c\ln n}|k\rangle\langle k|,\quad\tilde C=\sum_kC_{kk}|k\rangle\langle k|.
\end{equation}
Since $\tilde C$ is the diagonal part of $C$ (in the energy eigenbasis), it is easy to see $\|\tilde C\|\le\|C\|$. The approximation error in the first step of Eq. (\ref{eth}) is
\begin{multline}
\frac{1}{d}\left|\sum_j\sum_{k:|E_j-E_k|\ge c\ln n}A_{jj}B_{jk}C_{kk}D_{kj}\right|\le\frac{1}{d}\sum_j|A_{jj}|\left|\sum_{k:|E_j-E_k|\ge c\ln n}B_{jk}C_{kk}D_{kj}\right|\\
\le\frac{\|A\|}{d}\sum_j|\langle j|BQ_j\tilde CQ_jD|j\rangle|\le\frac{\|A\|}{d}\sum_j\|Q_jB^\dag|j\rangle\|\|\tilde C\|\|Q_jD|j\rangle\|\le\|A\|\|B\|\|C\|\|D\|/\poly n,
\end{multline}
where we used Lemma \ref{trash}. The approximation error in the third step of Eq. (\ref{eth}) can be upper bounded similarly.
\end{proof}

\begin{proposition} \label{p2}
The approximation error in the second step of Eq. (\ref{eth}) is $\tilde O(n^{-1.5})$.
\end{proposition}

\begin{proof}
Let $n$ be sufficiently large such that $n^{0.51}+c\ln n<\delta n$, and define
\begin{equation}
\tilde C^{(j)}:=\sum_{k:|E_j-E_k|<c\ln n}(C_{jj}-C_{kk})|k\rangle\langle k|.
\end{equation}
For $j,k$ such that $|E_j|<n^{0.51}$ and $|E_j-E_k|<c\ln n$, Assumption \ref{ethasmp} implies
\begin{equation}
|C_{jj}-C_{kk}|\le|f_C(E_j/n)-f_C(E_k/n)|+1/\poly n=O(|E_j-E_k|)/n+1/\poly n.
\end{equation}
Hence, $\|\tilde C^{(j)}\|=\tilde O(1/n)$ for any $j$ such that $|E_j|<n^{0.51}$. The approximation error in the second step of Eq. (\ref{eth}) is
\begin{align}
&\frac{1}{d}\left|\sum_j\sum_{k:|E_j-E_k|<c\ln n}A_{jj}B_{jk}(C_{jj}-C_{kk})D_{kj}\right|\le\frac{1}{d}\sum_j|A_{jj}|\left|\sum_{k:|E_j-E_k|<c\ln n}B_{jk}\tilde C^{(j)}_{kk}D_{kj}\right|\nonumber\\
&=\frac{1}{d}\sum_j|A_{jj}||\langle j|B\tilde C^{(j)}D|j\rangle|\le\frac{1}{d}\sum_j|A_{jj}|\left\|\tilde C^{(j)}\right\|=\frac{1}{d}\sum_{j:|E_j|<n^{0.51}}|A_{jj}|\left\|\tilde C^{(j)}\right\|\nonumber\\
&+\frac{1}{d}\sum_{j:|E_j|\ge n^{0.51}}|A_{jj}|\left\|\tilde C^{(j)}\right\|\le\frac{1}{d}\sum_{j:|E_j|<n^{0.51}}|A_{jj}|\tilde O(1/n)+\frac{1}{d}\sum_{j:|E_j|\ge n^{0.51}}|A_{jj}|O(1)\nonumber\\
&\le\frac{\tilde O(1/n)}{d}\sum_j|A_{jj}|+\frac{1}{d}\sum_{j:|E_j|\ge n^{0.51}}O(1)\le\tilde O(1/n)\sqrt{\frac{1}{d}\sum_j|A_{jj}|^2}+1/\poly n=\tilde O(n^{-1.5}),
\end{align}
where we used Eq. (\ref{exp}) in the last step.
\end{proof}

\begin{proof} [Justification of Eq. (\ref{theory})]
Specializing to $\langle AB(t)A^\dag B^\dag(t)\rangle$, the derivation above yields
\begin{equation} \label{abab}
{\rm OTOC}_\infty(A,B,A^\dag,B^\dag)=\frac{1}{d}\sum_j(AA^\dag)_{jj}|B_{jj}|^2+|A_{jj}|^2(BB^\dag)_{jj}-|A_{jj}B_{jj}|^2+\tilde O(n^{-1.5}).
\end{equation}
Consider the first term on the right-hand side:
\begin{multline} \label{est}
\frac{1}{d}\sum_j(AA^\dag)_{jj}|B_{jj}|^2\approx\frac{1}{d}\sum_{j:|E_j|<n^{0.51}}(AA^\dag)_{jj}|B_{jj}|^2\approx\frac{1}{d}\sum_{j:|E_j|<n^{0.51}}f_{AA^\dag}(0)|B_{jj}|^2\\
\approx\frac{f_{AA^\dag}(0)}{d}\sum_j|B_{jj}|^2\approx\frac{\tr(AA^\dag)|\tr(HB)|^2}{d^2n\tr(HH_i)},
\end{multline}
where we used Lemma \ref{Mar} in the first and third steps; the continuity of $f_{AA^\dag}(x)$ at $x=0$ in the second step; Eqs. (\ref{zero}), (\ref{exp}) in the last step. A rigorous error analysis for Eq. (\ref{est}) is given in Proposition \ref{p3} below.

The second term on the right-hand side of Eq. (\ref{abab}) can be estimated similarly. The third term on the right-hand side of Eq. (\ref{abab}) is
\begin{equation}
\frac{1}{d}\sum_j|A_{jj}B_{jj}|^2\le\frac{1}{2d}\sum_j|A_{jj}|^4+|B_{jj}|^4=O(n^{-2}),
\end{equation}
where we used Eq. (\ref{exp2}). Thus, Eq. (\ref{theory}) is proved based on Assumptions \ref{generic}, \ref{ethasmp}.
\end{proof}

\begin{proposition} \label{p3}
The error in Eq. (\ref{est}) is $O(n^{-2})$.
\end{proposition}

\begin{proof}
The approximation error in the last step of Eq. (\ref{est}) is $O(n^{-2})$ as given by Eq. (\ref{exp}). Using the Taylor expansion of $f_{AA^\dag}(x)$ at $x=0$, we estimate the approximation error in the second step of Eq. (\ref{est}):
\begin{align}
&\frac{O(1)}{d}\left|\sum_{j:|E_j|<n^{0.51}}\frac{|B_{jj}|^2E_j}{n}\right|+\frac{O(1)}{d}\sum_{j:|E_j|<n^{0.51}}\frac{|B_{jj}|^2E_j^2}{n^2}+1/\poly n\lesssim\frac{O(1)}{d}\left|\sum_{j:|E_j|<n^{0.51}}\frac{E_j^3}{n^3}\right|\nonumber\\
&+\frac{O(1)}{d}\sum_{j:|E_j|<n^{0.51}}\left(\frac{E_j^4}{n^4}+|B_{jj}|^4\right)\approx\frac{O(1)}{d}\left|\sum_j\frac{E_j^3}{n^3}\right|+\frac{O(1)}{d}\sum_j\left(\frac{E_j^4}{n^4}+|B_{jj}|^4\right)\nonumber\\
&\approx O(n^{-2})+O(n^{-2})+O(n^{-2})=O(n^{-2}),
\end{align}
where we used the Taylor expansion of $f_B(x)$ at $x=0$ and the inequality of arithmetic and geometric means in the first step; Lemma \ref{Mar} in the second step; Eqs. (\ref{17}), (\ref{18}), (\ref{exp2}) in the third step.
\end{proof}

\section{Chaotic dynamics as random unitary} \label{randomU}

In this section, we rederive Eq. (\ref{spec}) using techniques from the theory of random unitaries. The derivation is not rigorous, but provides a heuristic picture showing the extent to which chaotic dynamics can be approximated by a random unitary.

To improve the approximation described by Eq. (\ref{At}), we first take into account the unitarity of the dynamics. In strongly chaotic systems, it is tempting to expect
\begin{assumption} \label{Kitaev}
The time evolution operator $e^{-iHt}$ for large $t$ behaves like a random unitary.
\end{assumption}
Based on this assumption, late-time OTOC can be estimated from
\begin{equation}
{\rm OTOC}_\infty(A,B,C,D)=\int\mathrm dU\langle A(U^\dag BU)C(U^\dag DU)\rangle,
\end{equation}
where $U$ is taken from the unitary group $\mathcal U(d)$ with respect to the Haar measure.

\begin{lemma} [\cite{Kit16, RY16}] \label{Haar}
\begin{equation} \label{scramble}
\int\mathrm dU\langle AU^\dag BUCU^\dag DU\rangle=\langle A,B,C,D\rangle-\frac{\langle AC\rangle_c\langle BD\rangle_c}{d^2-1},
\end{equation}
where $\langle XY\rangle_c:=\langle XY\rangle-\langle X\rangle\langle Y\rangle$ is the connected correlator and
\begin{equation} \label{connected}
\langle A,B,C,D\rangle:=\langle AC\rangle\langle B\rangle\langle D\rangle+\langle A\rangle\langle C\rangle\langle BD\rangle-\langle A\rangle\langle B\rangle\langle C\rangle\langle D\rangle.
\end{equation}
\end{lemma}

Note that the right-hand side of Eq. (\ref{eigenconn}) resembles that of Eq. (\ref{connected}) in the sense of replacing every $\langle\cdots\rangle$ (expectation value at infinite temperature) by $\langle j|\cdots|j\rangle$ (expectation value in an eigenstate).

\begin{corollary} [\cite{Kit16, RY16}] \label{oto}
Assumption \ref{Kitaev} and Lemma \ref{Haar} imply
\begin{equation} \label{full}
{\rm OTOC}_\infty(A,B,C,D)=\langle A,B,C,D\rangle-\frac{\langle AC\rangle_c\langle BD\rangle_c}{d^2-1}.
\end{equation}
Therefore, 
\begin{itemize}
\item ${\rm OTOC}_\infty(A,B,A^\dag,B^\dag)$ for traceless operators $A,B$ vanishes in the thermodynamic limit $n\to\infty$. 
\item In finite-size systems, the saturation value of OTOC $\langle AB(t)A^\dag B^\dag(t)\rangle$ is exponentially small in the system size (because $d=2^n$).
\end{itemize}
\end{corollary}

The approximation stated in Assumption \ref{Kitaev} is still too crude. We propose a refined version of Assumption \ref{Kitaev} by incorporating energy conservation and argue (nonrigorously) that Eq. (\ref{spec}) follows from this refinement. 

We observe that the time evolution conserves energy and that local operators can only additively change the energy of a state by $O(1)$ (Lemma \ref{trash}). Thus, the action of OTOC $AB(t)CD(t)$ is approximately restricted to each microcanonical ensemble. This observation motivates a refinement of Assumption \ref{Kitaev} in strongly chaotic systems:

\begin{assumption} \label{Huang}
The time evolution operator $e^{-iHt}$ for large $t$ behaves like a random unitary in each microcanonical ensemble.
\end{assumption}

Conceptually, this assumption is related to the so-called random diagonal unitaries \cite{NTM12, NHKW17}.

Based on Assumption \ref{Huang}, we argue for Eq. (\ref{spec}). Since the bandwidth of $H$ is $\Theta(n)$, we decompose the energy spectrum into a disjoint union of $\Theta(n/\Delta)$ microcanonical ensembles with bandwidth $\Delta$. Let $J_k:=[k\Delta,(k+1)\Delta)$ and define $[A,B,C,D]_k$ as the right-hand side of Eq. (\ref{connected}) with every $\langle\cdots\rangle$ replaced by the expectation value $\tr(P_{J_k}\cdots)/\tr P_{J_k}$ in the microcanonical ensemble. We expect
\begin{equation} \label{chain}
\lim_{\tau\to\infty}\frac{1}{\tau}\int_0^\tau\mathrm dt\frac{\tr(P_{J_k}AB(t)CD(t))}{\tr P_{J_k}}\approx[A,B,C,D]_k\approx\frac{1}{\tr P_{J_k}}\sum_{j:E_j\in J_k}\langle A,B,C,D\rangle_j.
\end{equation}
The first step is a consequence of Lemma \ref{Haar} and Assumption \ref{Huang}. Indeed, it is just Eq. (\ref{full}) restricted to the microcanonical ensemble $P_{J_k}$. The last step of Eq. (\ref{chain}) used ETH. Equation (\ref{spec}) follows immediately from Eq. (\ref{chain}).

An important subtlety here, which does not appear in the derivation of Eq. (\ref{full}), requires further explanation. For an eigenstate $|j\rangle$ in a microcanonical ensemble $P_{J_k}$, the state $AB(t)CD(t)|j\rangle$ may not be completely in the microcanonical ensemble. As long as $A,B,C,D$ are local operators, Lemma \ref{trash} implies
\begin{equation} \label{trunc}
\|(1-P_{J_k})AB(t)CD(t)|j\rangle\|\le\|A\|\|B\|\|C\|\|D\|e^{-\Omega(\min\{E_j-k\Delta,(k+1)\Delta-E_j\})},
\end{equation}
i.e., the ``leakage'' out of the microcanonical ensemble $P_{J_k}$ is exponentially small. This is why Eq. (\ref{spec}) requires the locality of $A,B,C,D$, although Corollary \ref{oto} does not.

\section{Numerics} \label{numerics}

In this section, we support Eq. (\ref{theory}) with numerical simulations. Consider the spin-$1/2$ chain
\begin{equation} \label{hastings}
H=\sum_{i=1}^nH_i,\quad H_i=\sigma_i^z\sigma_{i+1}^z-1.05\sigma_i^x+0.5\sigma_i^z+g\sigma_i^y\sigma_{i+1}^z
\end{equation}
with periodic boundary conditions ($\sigma_{n+1}^z:=\sigma_1^z$), where $\sigma_i^x,\sigma_i^y,\sigma_i^z$ are the Pauli matrices at site $i$. For $g=0$, this model is nonintegrable in the sense of Wigner-Dyson level statistics \cite{BCH11, KH13}. Reference \cite{RSS15} calculated OTOC, focusing on the butterfly effect rather than the late-time behavior. Note that for $g=0$, most energy levels are twofold degenerate so that Assumption \ref{generic} does not hold.

We fix $g=0.1$. Intuitively, the model is nonintegrable for any value of $g$. We have numerically confirmed the validity of Assumption \ref{generic} for $n=5,6,\ldots,12$. Presumably, Assumption \ref{generic} holds for any integer $n\ge5$. Let
\begin{equation}
F_n^x:={\rm OTOC}_\infty(\sigma_1^x,\sigma_i^x,\sigma_1^x,\sigma_i^x),\quad F_n^z:={\rm OTOC}_\infty(\sigma_1^z,\sigma_i^z,\sigma_1^z,\sigma_i^z).
\end{equation}
Note that the values of $F_n^x,F_n^z$ are independent of $i$. We compute $F_n^x,F_n^z$ using exact diagonalization. The results are shown in the left panel of Fig. \ref{plot}.

\begin{figure}
\includegraphics[width=0.5\linewidth]{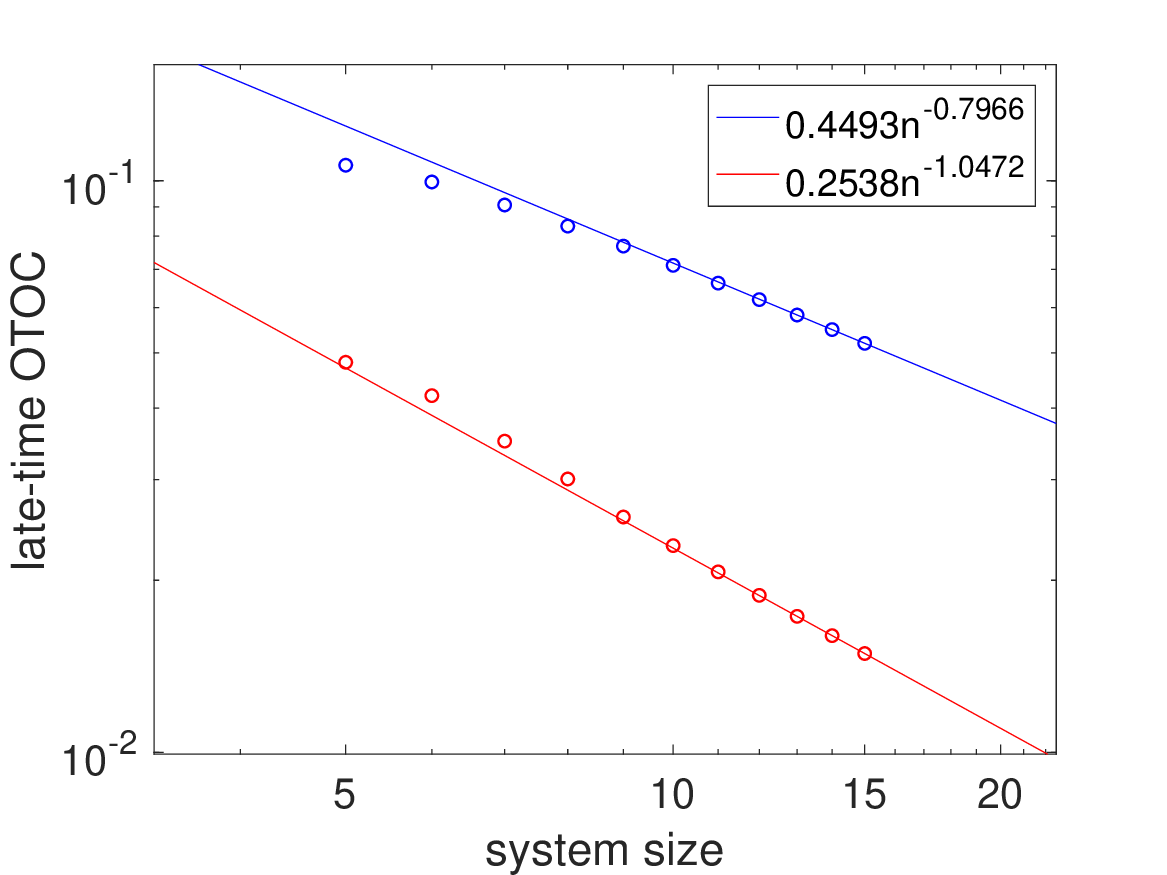}
\includegraphics[width=0.5\linewidth]{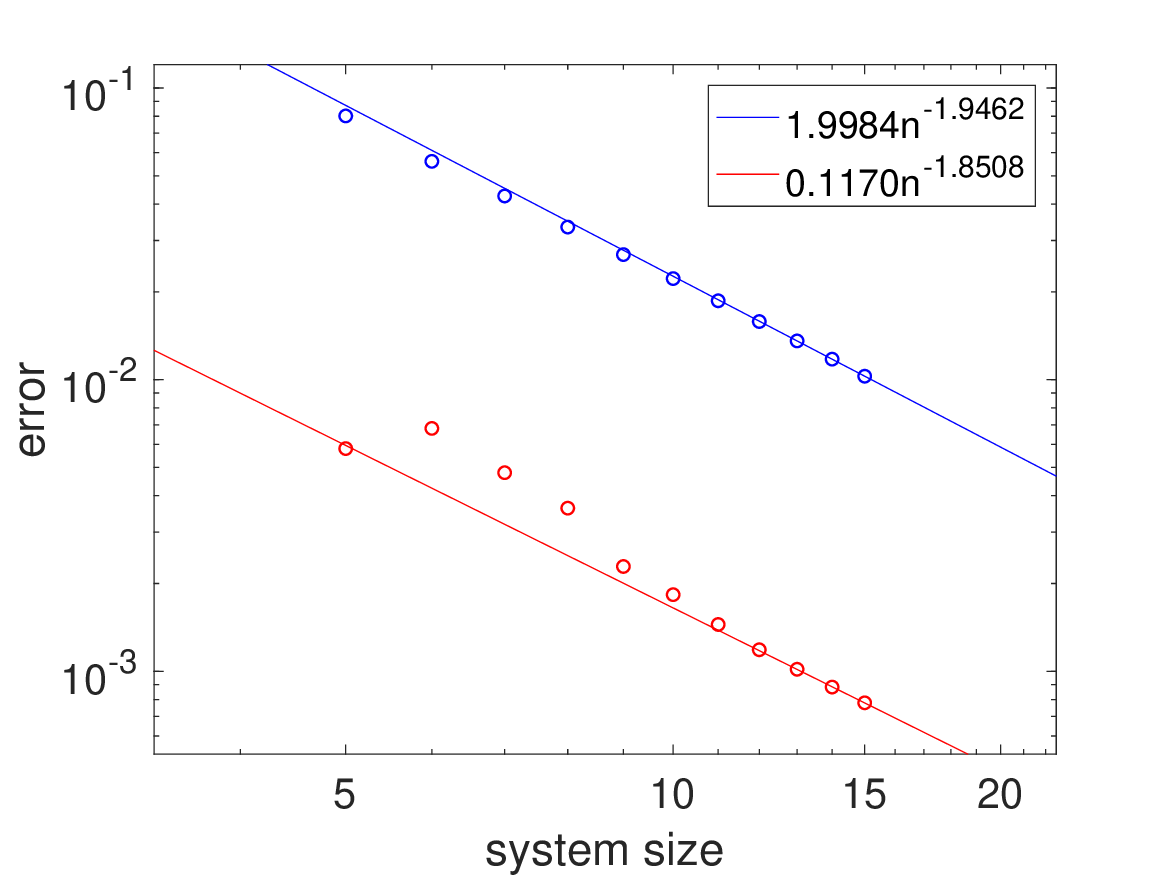}
\caption{Left panel: Finite-size scaling of late-time OTOC $F_n^x$ (blue), $F_n^z$ (red) for $n=5,6,\ldots,15$. The lines are power-law fits $0.4493n^{-0.7966}$ (blue), $0.2538n^{-1.0472}$ (red) to the last few data points. Right panel: Finite-size scaling of the errors $|F_n^x-G_n^x|$ (blue), $|F_n^z-G_n^z|$ (red) for $n=5,6,\ldots,15$. The lines are power-law fits $1.9984n^{-1.9462}$ (blue), $0.1170n^{-1.8508}$ (red) to the last few data points.} \label{plot}
\end{figure}

The leading terms in the finite-size scaling of $F_n^x,F_n^z$ are calculated analytically from Eq. (\ref{theory}):
\begin{equation}
G_n^x:=\frac{14}{15n}\approx\frac{0.9333}{n},\quad G_n^z:=\frac{40}{189n}\approx\frac{0.2116}{n}.
\end{equation}
We expect that the noticeable differences between $G_n^x,G_n^z$ and the power-law fits to $F_n^x,F_n^z$ are due to finite-size effects. To justify this claim, we perform a scaling analysis of the errors $|F_n^x-G_n^x|$, $|F_n^z-G_n^z|$ in the right panel of Fig. \ref{plot}. The numerics suggest that the errors should vanish as $\Theta(n^{-2})$ in the thermodynamic limit $n\to\infty$.

\section{Conclusion}

We propose that in order to better approximate the late-time behavior of chaotic dynamics generated by a time-independent Hamiltonian, one needs to take into account energy conservation. In particular, we show that approximation schemes with and without energy conservation make different predictions about OTOC at late times: without energy conservation, late-time OTOC scales inverse exponentially with system size; with energy conservation, the scaling is inverse polynomial. The latter prediction has been rigorously confirmed based on two very mild assumptions and is consistent with numerical simulations of a nonintegrable spin chain. 

An immediate open question is how good the energy-preserving approximation scheme proposed in this paper is in predicting the late-time behavior of higher-order time-ordered or out-of-time-ordered correlators. A more general problem for future study is how to approximate the time evolution process and capture other universal features of chaotic dynamics. See Refs. \cite{KVH18, RPv18, RPv19, Hua19} for recent progress in this direction.

\section*{Acknowledgments}

We would like to thank Xie Chen, Yingfei Gu, Nick Hunter-Jones, Alexei Y. Kitaev, and Douglas Stanford for helpful discussions. We are especially grateful to X.C., who wrote a substantial portion of this paper. We acknowledge funding provided by the Institute for Quantum Information and Matter, an NSF Physics Frontiers Center (NSF Grant PHY-1733907). Additional funding support was provided by NSF DMR-1654340 (Y.H., Y.-L.Z.).

\bibliographystyle{abbrv}
\bibliography{otoc}

\begin{thebibliography}{10}

\bibitem{AKL16}
I.~Arad, T.~Kuwahara, and Z.~Landau.
\newblock Connecting global and local energy distributions in quantum spin
  models on a lattice.
\newblock {\em Journal of Statistical Mechanics: Theory and Experiment},
  2016(3):033301, 2016.

\bibitem{BCH11}
M.~C. Ba\~nuls, J.~I. Cirac, and M.~B. Hastings.
\newblock Strong and weak thermalization of infinite nonintegrable quantum
  systems.
\newblock {\em Physical Review Letters}, 106(5):050405, 2011.

\bibitem{BC15}
F.~G. S.~L. Brandao and M.~Cramer.
\newblock Equivalence of statistical mechanical ensembles for non-critical
  quantum systems.
\newblock arXiv:1502.03263, 2015.

\bibitem{CZHF17}
X.~Chen, T.~Zhou, D.~A. Huse, and E.~Fradkin.
\newblock Out-of-time-order correlations in many-body localized and thermal
  phases.
\newblock {\em Annalen der Physik}, 529(7):1600332, 2017.

\bibitem{Che16}
Y.~Chen.
\newblock Universal logarithmic scrambling in many body localization.
\newblock arXiv:1608.02765, 2016.

\bibitem{Deu91}
J.~M. Deutsch.
\newblock Quantum statistical mechanics in a closed system.
\newblock {\em Physical Review A}, 43(4):2046--2049, 1991.

\bibitem{FZSZ17}
R.~Fan, P.~Zhang, H.~Shen, and H.~Zhai.
\newblock Out-of-time-order correlation for many-body localization.
\newblock {\em Science Bulletin}, 62(10):707--711, 2017.

\bibitem{GQS17}
Y.~Gu, X.-L. Qi, and D.~Stanford.
\newblock Local criticality, diffusion and chaos in generalized
  {S}achdev-{Y}e-{K}itaev models.
\newblock {\em Journal of High Energy Physics}, 2017(5):125, 2017.

\bibitem{HK06}
M.~B. Hastings and T.~Koma.
\newblock Spectral gap and exponential decay of correlations.
\newblock {\em Communications in Mathematical Physics}, 265(3):781--804, 2006.

\bibitem{HL17}
R.-Q. He and Z.-Y. Lu.
\newblock Characterizing many-body localization by out-of-time-ordered
  correlation.
\newblock {\em Physical Review B}, 95(5):054201, 2017.

\bibitem{HQRS16}
P.~Hosur, X.-L. Qi, D.~A. Roberts, and B.~Yoshida.
\newblock Chaos in quantum channels.
\newblock {\em Journal of High Energy Physics}, 2016(2):4, 2016.

\bibitem{Hua19}
Y.~Huang.
\newblock Dynamics of {R}enyi entanglement entropy in local quantum circuits
  with charge conservation.
\newblock arXiv:1902.00977, 2019.

\bibitem{HZC17}
Y.~Huang, Y.-L. Zhang, and X.~Chen.
\newblock Out-of-time-ordered correlators in many-body localized systems.
\newblock {\em Annalen der Physik}, 529(7):1600318, 2017.

\bibitem{KLW15}
J.~P. Keating, N.~Linden, and H.~J. Wells.
\newblock Spectra and eigenstates of spin chain {H}amiltonians.
\newblock {\em Communications in Mathematical Physics}, 338(1):81--102, 2015.

\bibitem{KVH18}
V.~Khemani, A.~Vishwanath, and D.~A. Huse.
\newblock Operator spreading and the emergence of dissipative hydrodynamics
  under unitary evolution with conservation laws.
\newblock {\em Physical Review X}, 8(3):031057, 2018.

\bibitem{KH13}
H.~Kim and D.~A. Huse.
\newblock Ballistic spreading of entanglement in a diffusive nonintegrable
  system.
\newblock {\em Physical Review Letters}, 111(12):127205, 2013.

\bibitem{KIH14}
H.~Kim, T.~N. Ikeda, and D.~A. Huse.
\newblock Testing whether all eigenstates obey the eigenstate thermalization
  hypothesis.
\newblock {\em Physical Review E}, 90(5):052105, 2014.

\bibitem{Kit16}
A.~Kitaev.
\newblock Averaging over the unitary group.
\newblock Ph/CS219C (Spring 2016), Caltech, course note.

\bibitem{Kit14}
A.~Kitaev.
\newblock Hidden correlations in the {H}awking radiation and thermal noise.
\newblock In {\em 2015 Breakthrough Prize in Fundamental Physics Symposium},
  2014.

\bibitem{Kit15}
A.~Kitaev.
\newblock A simple model of quantum holography.
\newblock In {\em KITP Program: Entanglement in Strongly-Correlated Quantum
  Matter}, 2015.

\bibitem{LO69}
A.~I. Larkin and Y.~N. Ovchinnikov.
\newblock Quasiclassical method in the theory of superconductivity.
\newblock {\em Sov. Phys. JETP}, 28(6):1200--1205, 1969.

\bibitem{LR72}
E.~H. Lieb and D.~W. Robinson.
\newblock The finite group velocity of quantum spin systems.
\newblock {\em Communications in Mathematical Physics}, 28(3):251--257, 1972.

\bibitem{MSS16}
J.~Maldacena, S.~H. Shenker, and D.~Stanford.
\newblock A bound on chaos.
\newblock {\em Journal of High Energy Physics}, 2016(8):106, 2016.

\bibitem{MS16}
J.~Maldacena and D.~Stanford.
\newblock Remarks on the {S}achdev-{Y}e-{K}itaev model.
\newblock {\em Physical Review D}, 94(10):106002, 2016.

\bibitem{NS06}
B.~Nachtergaele and R.~Sims.
\newblock Lieb-{R}obinson bounds and the exponential clustering theorem.
\newblock {\em Communications in Mathematical Physics}, 265(1):119--130, 2006.

\bibitem{NHKW17}
Y.~Nakata, C.~Hirche, M.~Koashi, and A.~Winter.
\newblock Efficient quantum pseudorandomness with nearly time-independent
  {H}amiltonian dynamics.
\newblock {\em Physical Review X}, 7(2):021006, 2017.

\bibitem{NTM12}
Y.~Nakata, P.~S. Turner, and M.~Murao.
\newblock Phase-random states: Ensembles of states with fixed amplitudes and
  uniformly distributed phases in a fixed basis.
\newblock {\em Physical Review A}, 86(1):012301, 2012.

\bibitem{RPv18}
T.~Rakovszky, F.~Pollmann, and C.~W. von Keyserlingk.
\newblock Diffusive hydrodynamics of out-of-time-ordered correlators with
  charge conservation.
\newblock {\em Physical Review X}, 8(3):031058, 2018.

\bibitem{RPv19}
T.~Rakovszky, F.~Pollmann, and C.~W. von Keyserlingk.
\newblock Sub-ballistic growth of {R}\'enyi entropies due to diffusion.
\newblock arXiv:1901.10502, 2019.

\bibitem{RDO08}
M.~Rigol, V.~Dunjko, and M.~Olshanii.
\newblock Thermalization and its mechanism for generic isolated quantum
  systems.
\newblock {\em Nature}, 452(7):854--858, 2008.

\bibitem{RS15}
D.~A. Roberts and D.~Stanford.
\newblock Diagnosing chaos using four-point functions in two-dimensional
  conformal field theory.
\newblock {\em Physical Review Letters}, 115(13):131603, 2015.

\bibitem{RSS15}
D.~A. Roberts, D.~Stanford, and L.~Susskind.
\newblock Localized shocks.
\newblock {\em Journal of High Energy Physics}, 2015(3):51, 2015.

\bibitem{RS16}
D.~A. Roberts and B.~Swingle.
\newblock Lieb-{R}obinson bound and the butterfly effect in quantum field
  theories.
\newblock {\em Physical Review Letters}, 117(9):091602, 2016.

\bibitem{RY16}
D.~A. Roberts and B.~Yoshida.
\newblock Chaos and complexity by design.
\newblock {\em Journal of High Energy Physics}, 2017(4):121, 2017.

\bibitem{SS14}
S.~H. Shenker and D.~Stanford.
\newblock Multiple shocks.
\newblock {\em Journal of High Energy Physics}, 2014(12):46, 2014.

\bibitem{SS15}
S.~H. Shenker and D.~Stanford.
\newblock Stringy effects in scrambling.
\newblock {\em Journal of High Energy Physics}, 2015(5):132, 2015.

\bibitem{Sre94}
M.~Srednicki.
\newblock Chaos and quantum thermalization.
\newblock {\em Physical Review E}, 50(2):888--901, 1994.

\bibitem{Sre99}
M.~Srednicki.
\newblock The approach to thermal equilibrium in quantized chaotic systems.
\newblock {\em Journal of Physics A: Mathematical and General},
  32(7):1163--1175, 1999.

\bibitem{SBSH16}
B.~Swingle, G.~Bentsen, M.~Schleier-Smith, and P.~Hayden.
\newblock Measuring the scrambling of quantum information.
\newblock {\em Physical Review A}, 94(4):040302, 2016.

\bibitem{SC17}
B.~Swingle and D.~Chowdhury.
\newblock Slow scrambling in disordered quantum systems.
\newblock {\em Physical Review B}, 95(6):060201, 2017.

\bibitem{ZHC19}
Y.-L. Zhang, Y.~Huang, and X.~Chen.
\newblock Information scrambling in chaotic systems with dissipation.
\newblock {\em Physical Review B}, 99(1):014303, 2019.

\end{thebibliography}

\end{document}